\documentclass[a4paper,prl,twocolumn,showpacs,superscriptaddress,groupedaddress]{revtex4}
\usepackage{ae}
\usepackage{physics}
\usepackage[T1]{fontenc}
\usepackage[ansinew]{inputenc}
\usepackage{amsmath}
\usepackage{amssymb}
\usepackage[caption=false]{subfig}
\usepackage{multirow}
\usepackage{array}
\usepackage[]{graphicx}
\usepackage{wrapfig}
\usepackage{makecell}
\usepackage{epstopdf}
\usepackage{color}
\usepackage[colorlinks]{hyperref}
\usepackage{lscape}
\usepackage{amsthm}
\newtheorem{theorem}{Theorem}
\graphicspath{ {./images1/} }

\begin{document}
\title{Entangled state distillation from single copy mixed states beyond LOCC}
\author{Indranil Biswas}
\email{indranilbiswas74@gmail.com}
\affiliation{Department of Applied Mathematics, University of Calcutta, 92, A.P.C. Road, Kolkata- 700009, India}
\author{Atanu Bhunia}
\email{atanu.bhunia31@gmail.com}
\affiliation{Department of Applied Mathematics, University of Calcutta, 92, A.P.C. Road, Kolkata- 700009, India}
\author{Indrani Chattopadhyay}
\email{icappmath@caluniv.ac.in}
\affiliation{Department of Applied Mathematics, University of Calcutta, 92, A.P.C. Road, Kolkata- 700009, India}
\author{Debasis Sarkar}
\email{dsarkar1x@gmail.com, dsappmath@caluniv.ac.in}
\affiliation{Department of Applied Mathematics, University of Calcutta, 92, A.P.C. Road, Kolkata- 700009, India}
\begin{abstract}
   No pure entangled state can be distilled from a $2\otimes 2$ or $2\otimes 3$ mixed state by separable operations. In $3\otimes 3$, pure entanglement can be distilled by separable operation but not by LOCC. In this letter, we proved the conjecture [PRL. 103, 110502 (2009)] that it is possible to distill pure entanglement for $2\otimes 4$ system by LOCC and further improve these in higher dimensions to distill a pure entangled state of Schmidt rank $d$ from a $m\otimes n$ mixed state by separable operation when $m+n \geqslant 3d$. We found results for tripartite systems with target state $d$-level GHZ-type state. These results provide a class of systems where separable operation is strictly stronger than LOCC.
\end{abstract}
\date{\today}
\pacs{03.67.Mn.; 03.65.Ud.; Keywords: Entanglement, LOCC, Separable operations, distillation;}
\maketitle

The phenomenon of entanglement in composite quantum systems is a striking feature without any analogue in classical physics. Although initially studied to unreveal the foundational aspects, the progress made in the last few decades elevated it to a resource in quantum information theory. Quantum teleportation, superdense coding, quantum key distribution, etc., are some of the examples of quantum information processing tasks where entanglement plays a crucial role. The precision of such tasks depend on the presence of a pure entangled state, preferably a maximally entangled state of a bipartite system. But such states are rare in nature. Due to imperfect quantum operations and decoherence, we usually have a noisy state in the laboratory. Bennett, Brassard et. al. \cite{1} in 1996 first considered this problem and developed a protocol to distill $k$ copies of nearly singlet states from $n(>k)$ copies of noisy entangled states. However, `useful' entanglement cannot be distilled from every noisy entangled state. There are examples of bound entangled states \cite{2} from which pure entangled states cannot be distilled even in the asymptotic regime.\\

Pure entangled states are very fragile in nature. Entanglement is lost whenever the state is exposed to excessive environmental noise. A clever way to perform quantum operations is to measure one party at a time. The outcome is then broadcasted globally by that party through classical means of communications. Depending on the outcome, some other parties will then measure his/her local systems and again broadcast the outcomes globally. This continues until the termination of that particular protocol. Such measurement schemes are acronymed by LOCC (Local Operations along with Classical Communications) \cite{3} in the literature. Despite having a simple physical description, mathematical characterization of LOCC is yet to discover. Nevertheless, a larger class of operations than LOCC, called separable operations have a rather compact mathematical structure.\\

Every physically realizable quantum operation is described by a completely positive superoperator  $\Omega$. When applied on a physical state $\rho$,  such superoperators map $\rho$ by the rule $\Omega(\rho)= \sum_{i=1}^n {E_i}{\rho}{E_i}^{\dagger}$ satisfying $\sum_{i=1}^n {E_i}^{\dagger}{E_i}= \mathbb{I}$, where $\{{{E_i\}}_{i=1}^n}$ is called a complete set of Kraus operators \cite{4}. If each ${E_i}$'s are of the form ${E_i}={A_i}^1 \otimes{A_i}^2 \otimes........\otimes{A_i}^k$ for a $k$-partite system, then the operation is called a separable operation (SEP). It is evident from the mathematical description of separable operation that LOCC forms a strict subclass of the former. Although, LOCC $\subsetneq$ SEP, it is still intriguing to know exactly where LOCC differs from separable operation. Several distinction has been found in terms of distinguishability of states \cite{5,6,7,8,9,10,11,12,13,14,15,16,17,18,19}, of these two classes of operation. Another arena where this distinction exists is deterministic transformation of single copy mixed states into pure entangled states \cite{21}. Kent \cite{20} and Chitamber et. al. \cite{21} have shown that both these two classes of operation are equally incompetent in distilling pure entangled states from one copy of a mixed state for $2\otimes 2$ and $2\otimes 3$ systems. But for $3\otimes 3$ systems, separable operation outperforms LOCC. These results inspired us to ask a very relevant question: Does there exist any relation between distillation of pure entangled states by separable operations and dimensions of the subsystems? Also, how LOCC and separable operations differ under general distillation scenario? We have found some affirmative answers to these questions for bipartite quantum systems which are considered as long standing problem in quantum information theory. We have further improved these results upto some extent for tripartite systems. Our result generalizes the possibility of distilling a pure entangled state of Schmidt rank $d$ from a bipartite $m\otimes n$ system by separable operation when $m+n\geqslant 3d$. We have constructed a generic class of mixed states and a class of separable operators that can distill a pure maximally entangled state when the equality holds. Furthermore, when one of the subsystem is of dimension $2d$ (and the other is of dimension at least $d$), the transformation is even possible by LOCC and thereby achieving a class of bipartite systems where separable operation is strictly stronger than LOCC. Note that in this letter we are only concerned about deterministic distillation \cite{21}. Unless otherwise stated, the term distillation would simply mean here distillation in deterministic sense.\\

For a tripartite $p\otimes q\otimes r$ system, we have upgraded the bipartite bound to $p+q+r\geqslant 4d$. However, this bound exists only for certain classes of mixed states and when the target state is a $d$-level Greenberger-Horne-Zeilinger (GHZ) type state where $\ket{GHZ_d}=\frac{1}{\sqrt{d}}\sum_{i=0}^{d-1}\ket{i}\otimes\ket{i}\otimes\ket{i}$. A distinction of separable and LOCC classes, similar to that of bipartite system has also been found. We now move to present the main results of this letter below.

In \cite{21}, it is conjectured that there exist mixed states in $2\otimes n$ $(n>3)$ systems from which pure entangled states can be distilled. In fact, LOCC is sufficient to distill such a state. We find that this is indeed the case. Consider the state of a $2\otimes 4$ system of the form $\rho=p\ket{\psi_1}\bra{\psi_1}+(1-p)\ket{\psi_2}\bra{\psi_2}$ shared between Alice and Bob, where $\ket{\psi_1}=\frac{1}{\sqrt{2}}(\ket{00}+\ket{11})$ and $\ket{\psi_2}=\frac{1}{\sqrt{2}}(\ket{02}+\ket{13}).$ Now, to distill the pure entangled state $\ket{\psi_1}$, Bob will perform a local measurement $\{B_1, B_2\}$ on his subsystem such that $B_1=\ket{0}\bra{0}+\ket{1}\bra{1}$, $B_2=\ket{0}\bra{2}+\ket{1}\bra{3}$.\\

Next, we observe that although LOCC and separable operations do not coincide in $3\otimes 3$ systems, the entangled state distilled by the latter class of operations is only of Schmidt rank $2$. Thus we have succeeded in characterization of bipartite systems for arbitrary (but finite) dimensions.\\
\begin{theorem}
There exist $m\otimes n$ mixed states that can be deterministically converted into a maximally entangled state of Schmidt rank $d$ by a suitable separable operation if and only if $m+n\geqslant 3d$ where $\min\{m,n\}\geqslant d$. Additionally, if $ \max\{m,n\}\geqslant 2d$ the task can also be achieved by LOCC.
\end{theorem}
\begin{proof}

Without loss of generality, we assume that $n \geqslant m \geqslant d$ and $m=d+k_1, n = d+k_2$ such that $k_1 +k_2 = d$ where $k_1\geqslant 1, k_2\geqslant 1$ are integers so that $m+n = 3d$. Consider the states $\ket{\psi_1}=\frac{1}{\sqrt{d}}\sum_{i=0}^{d-1}\ket{i,i}$ and $\ket{\psi_2}=\frac{1}{\sqrt{d}}\sum_{i=0}^{d-1}\ket{k_1 +i,d\oplus_n i}$ where $\oplus_k$ denotes the binary operation addition modulo $k$ and $\ket{i,j}=\ket{i}\otimes \ket{j}$. Set, $\rho = p\ket{\psi_1}\bra{\psi_1}+(1-p)\ket{\psi_2}\bra{\psi_2}$ where $p\in(0,1)$. Since $\braket{\psi_1}{\psi_2}= 0$, $\rho$ is a valid mixed state of rank two. Then the separable operation with Kraus operators $\{ E_1, E_2 \}$ defined by $$E_1 = \big({\sum_{i=0}^{d-1} \eta_i \ket{i}\bra{i}}\big)\otimes \big({\sum_{i=0}^{d-1} \nu_i \ket{i}\bra{i}}\big),$$ $$E_2 = \big({\sum_{i=0}^{d-1} \eta_{i}{'} \ket{i}\bra{{k_1}+i}}\big)\otimes \big({\sum_{i=0}^{d-1} \nu_{i}{'} \ket{i}\bra{d\oplus_n i}}\big)$$ where
\begin{align*}
\eta_i=\begin{cases} 1 \hspace{4mm} ;\: 0\leqslant i \leqslant k_1 -1 \\ \frac{1}{\sqrt{2}} \hspace{1mm} ;\: k_1 \leqslant i \leqslant d-1
\end{cases}
\nu_i=\begin{cases} \frac{1}{\sqrt{2}} \hspace{1mm} ;\: 0\leqslant i \leqslant k_1 -1 \\ 1 \hspace{4mm} ;\: k_1 \leqslant i \leqslant d-1
\end{cases}
\end{align*} \vspace{1mm}
\begin{align*}
\eta_{i}{'}=\begin{cases} \frac{1}{\sqrt{2}} \hspace{1mm} ;\: 0\leqslant i \leqslant k_2 -1 \\ 1 \hspace{4mm} ;\: k_2 \leqslant i \leqslant d-1
\end{cases}
\nu_{i}{'}=\begin{cases} 1 \hspace{4mm} ;\: 0\leqslant i \leqslant k_2 -1 \\ \frac{1}{\sqrt{2}} \hspace{1mm} ;\: k_2 \leqslant i \leqslant d-1
\end{cases}
\end{align*}
can convert $\rho$ into the pure maximally entangled state $\ket{\psi_1}$ of Schmidt rank $d$.\\

Also, if $m=d,\; n=2d$, then consider $\ket{\psi_1}=\frac{1}{\sqrt{d}}\sum_{i=0}^{d-1}\ket{i,i}$ and $\ket{\psi_2}=\frac{1}{\sqrt{d}}\sum_{i=0}^{d-1}\ket{i,d+i}$ such that $\rho= p\ket{\psi_1}\bra{\psi_1}+(1-p)\ket{\psi_2}\bra{\psi_2}$. Bob now applies a local measurement $\{B_1, B_2\}$ on his subsystem and converts $\rho$ into $\ket{\psi_1}$ where $B_1=\sum_{i=0}^{d-1} \ket{i}\bra{i}$ and $B_2=\sum_{i=0}^{d-1} \ket{i}\bra{d+i}$.\\

Conversely, we assume that $k_1+k_2< d$. If possible, suppose there exists a mixed state $\rho$ that can be converted into a pure entangled state of Schmidt rank $d$, by a suitable separable operation with Kraus operators $\{A_i\otimes B_i\}_i$. Without loss of generality, we assume that rank of $\rho$ to be two and $\ket{\psi_1}, \ket{\psi_2}$ are orthogonal eigenstates of $\rho$. The linear nature of separable operators implies that $\ket{\psi_1}, \ket{\psi_2}$ must both be of Schmidt rank at least $d$. So, there exists $A\otimes B\in\{A_i\otimes B_i\}_i$ such that $(A\otimes B)\ket{\psi_1}= a\ket{\phi}$ and $(A\otimes B)\ket{\psi_2}= b\ket{\phi}$ where $\ket{\phi}(a,b $ are scalars) is a pure entangled state of Schmidt rank $d$ and either (i) $a\neq 0, b\neq 0$ or (ii) only one of $a, b$ is non-zero. \\

For the case (i), $(A\otimes B)(b\ket{\psi_1}-a\ket{\psi_2})= 0$ implies that $b\ket{\psi_1}-a\ket{\psi_2}$ is a vector with Schmidt rank atmost $d-1$, which is impossible, since $(A\otimes B)(x\ket{\psi_1}+y\ket{\psi_2})= (ax+by)\ket{\phi}, \forall x,y \in \mathbb{C}$. For the case (ii), suppose $a\neq 0$. Then both $A$ and $B$ are operators of rank at least $d$ and therefore, for $b$ to be zero the sum of dimensions of the null spaces of $A$ and $B$ together is of at least $d$. But this contradicts our primary assumption that $k_1+k_2< d$.\\

Next for the case of LOCC, we follow the mathematical technique of \cite{21}. We assume that $\max\{k_1,k_2\}< d$. Without loss of generality, let $\rho$ be a mixed state of rank two and $\ket{\psi_1}, \ket{\psi_2}$ are its orthogonal eigenstates of Schmidt rank at least $d$. If possible, suppose there exists some LOCC protocol $\Omega$ such that $\Omega(\rho)= \ket{\phi}\bra{\phi}$ where $\ket{\phi}$ is a pure entangled state of Schmidt rank $d$. In fact, the LOCC protocol must transform $\ket{\psi_1}$ and $\ket{\psi_2}$ simultaneously into $\ket{\phi}$, i.e., $\Omega(\ket{\psi_1}\bra{\psi_1})= \Omega(\ket{\psi_2}\bra{\psi_2})=\ket{\phi}\bra{\phi}$. \\

Now, any bipartite LOCC protocol consists of alternating rounds of measurements and broadcasting of outcomes where one of the parties, say, Alice makes a measurement on her subsystem and communicates the result to the other party, say, Bob. He then chooses a measurement on his subsystem depending on the outcomes of Alice and reports back the outcome to Alice and the cycle continues in this way. These operations therefore create a tree of all possible outcomes where each branch represents a series of operators implemented by the two parties. Suppose $N_i$ index the $i$-th branch after $N$ rounds of measurement so that $A^{N_i}\otimes B^{N_i}$ describes the net operation throughout all $N$ rounds along $N_i$. We now prove the following claim: For any branch $N_i$, $A^{N_i}\otimes B^{N_i}\ket{\psi_1}\neq 0$ if and only if $A^{N_i}\otimes B^{N_i}\ket{\psi_2}\neq 0$.\\

By induction we assume that it is to be true and suppose any of Alice or Bob, say Bob, makes the next measurement. For any $B^j$ let $(N+1)_{ji}$ correspond to the branch $(\mathbb{I}_A \otimes B^j)(A^{N_i}\otimes B^{N_i})$. Since the conversion is deterministic, $A^{N_i}\otimes B^{N_i}\ket{\psi_1}$ and $A^{N_i}\otimes B^{N_i}\ket{\psi_2}$ are both states of Schmidt rank at least $d$. Thus if the measurement operator $B^j$ is applied, it must be of rank at least $d$ and therefore cannot annihilate any state with Schmidt rank $d$ or more. Hence, $A^{N_i}\otimes B^{N_i}\ket{\psi_1}\neq 0$ if and only if $A^{N_i}\otimes B^{N_i}\ket{\psi_2}\neq 0$. This argument also holds true when $N=0$ and therefore the claim is justified.\\

Next, consider a branch in the protocol that transforms $\ket{\psi_1}$ to $\ket{\phi}$. Clearly, $\ket{\psi_2}$ must also be converted into $\ket{\phi}$ along this branch. The states $\ket{\psi_1}$ and $\ket{\psi_2}$ are linearly independent, but  at the end they become the same state. So there exist some round $N$ on this branch such that $\ket{\psi_1^{N-1}}$ and $\ket{\psi_2^{N-1}}$ are linearly independent but $\ket{\psi_2^{N}}$ is a scalar multiple of $\ket{\psi_1^{N}}$. Here, $\ket{\psi_{1(2)}^{N}}$ is the resultant state after $N$ rounds, which was originally $\ket{\psi_{1(2)}}$. Without loss of generality, we assume that $N$-th round measurement is performed by Bob. Then $(\mathbb{I}_A \otimes B)\ket{\psi_1^{N-1}}=c\ket{\phi}$ and $(\mathbb{I}_A \otimes B)\ket{\psi_2^{N-1}}=d\ket{\phi},$ where $c, d$ are non-zero scalars. This implies that  $(\mathbb{I}_A \otimes B) (d\ket{\psi_1^{N-1}}-c\ket{\psi_2^{N-1}})= 0$ where $d\ket{\psi_1^{N-1}}-c\ket{\psi_2^{N-1}}$ is a pure entangled state of Schmidt rank at least $d$. But this contradicts the assumption that $k_2< d$.
\end{proof}

Notice that the class of separable operations can transform the class of mixed states into a maximally entangled state. Clearly this result is somewhat stronger than previous result in \cite{21}. However, even if the bound is satisfied by some bipartite system, there exist mixed states for which we could not find a suitable separable operation to convert it into a pure entangled state. Consider the state $\rho= p\ket{\psi_1}\bra{\psi_1}+(1-p)\ket{\psi_2}\bra{\psi_2}$ where $\ket{\psi_1}=\frac{1}{\sqrt{2}}(\ket{00}+\ket{11})$ and $\ket{\psi_2}= \frac{1}{\sqrt{2}}(\ket{01}+\ket{10})$. This is indeed one such example. We now proceed further to analyze the case of tripartite quantum systems.\\

In tripartite scenario, we will only be concerned about distilling GHZ-type states. We find that for a certain class of mixed states, it is possible to distill GHZ-type states with unit probability. A bound similar to bipartite systems also exist for this case.\\
\begin{theorem}
There exists mixed states in $p\otimes q \otimes r$ system that can be deterministically converted into a pure $d$-level GHZ-type state by a suitable separable operation if and only if $p+q+r \geqslant 4d$ where $\min\{p,q,r\}\geqslant d$. Additionally, if $\max\{p,q,r\}\geqslant 2d$ the task can also be achieved by LOCC.
\end{theorem}
\begin{proof}
Without loss of generality, we assume that $r\geqslant q \geqslant p \geqslant d$ and $p= d+k_1, q= d+k_2, r= d+k_3$ such that $k_1+k_2+k_3= d$, where $k_1, k_2, k_3$ are non-negative integers so that $p+q+r= 4d$. Consider a mixed state $\rho$ of this system defined by $\rho=p\ket{\psi_1}\bra{\psi_1}+(1-p)\ket{\psi_2}\bra{\psi_2}$ where $\ket{\psi_1}, \ket{\psi_2}$ are two orthogonal $d$-level(or more) GHZ-type state.\\

  (I) If $k_1= 0$ and $k_2, k_3$ are non-zero and suppose $$\ket{\psi_1}=\frac{1}{\sqrt{d}}\sum_{i=0}^{d-1}\ket{i,i,i}, \ket{\psi_2}=\frac{1}{\sqrt{d}}\sum_{i=0}^{d-1}\ket{i,k_2+i,d\oplus_{r}i}.$$ The separable operation $\{E_1, E_2\}$ defined by $$E_1 = \frac{1}{\sqrt{2}}\mathbb{I}_A \otimes \big({\sum_{i=0}^{d-1} \eta_i \ket{i}\bra{i}}\big)\otimes \big({\sum_{i=0}^{d-1} \nu_i \ket{i}\bra{i}}\big),$$ $$ E_2 = \frac{1}{\sqrt{2}}\mathbb{I}_A \otimes\big({\sum_{i=0}^{d-1} \eta_{i}{'} \ket{i}\bra{{k_2}+i}}\big)\otimes \big({\sum_{i=0}^{d-1} \nu_{i}{'} \ket{i}\bra{d\oplus_r i}}\big)$$ where 
  \begin{align*}
\eta_i=\begin{cases} 1 \hspace{6mm} ;\: 0\leqslant i \leqslant k_2 -1 \\ \frac{1}{\sqrt{2}} \hspace{3mm} ;\: k_2 \leqslant i \leqslant d-1
\end{cases}
\nu_i=\begin{cases} \frac{1}{\sqrt{2}} \hspace{3mm} ;\: 0\leqslant i \leqslant k_2 -1 \\ 1 \hspace{6mm} ;\: k_2 \leqslant i \leqslant d-1
\end{cases}
\end{align*} \vspace{1mm}
\begin{align*}
\eta_{i}{'}=\begin{cases} \frac{1}{\sqrt{2}} \hspace{3mm} ;\: 0\leqslant i \leqslant k_3 -1 \\ 1 \hspace{6mm} ;\: k_3 \leqslant i \leqslant d-1
\end{cases}
\nu_{i}{'}=\begin{cases} 1 \hspace{6mm} ;\: 0\leqslant i \leqslant k_3 -1 \\ \frac{1}{\sqrt{2}} \hspace{3mm} ;\: k_3 \leqslant i \leqslant d-1
\end{cases}
\end{align*}
(II) If all $k_i$'s are non-zero, let $$\ket{\psi_1}=\frac{1}{\sqrt{d}}\sum_{i=0}^{d-1}\ket{i,i,i},$$ $$\ket{\psi_2}= \frac{1}{\sqrt{d}}\sum_{i=0}^{d-1}\ket{k_1+i,\overline{k_1+k_2}\oplus_q i,d\oplus_r i}.$$ The separable operation $\{E_1,E_2\}$ defined by $$E_1 =\big({\sum_{i=0}^{d-1} \alpha_i \ket{i}\bra{i}}\big)  \otimes \big({\sum_{i=0}^{d-1} \beta_i \ket{i}\bra{i}}\big)\otimes \big({\sum_{i=0}^{d-1} \gamma_i \ket{i}\bra{i}}\big),$$ $$ E_2 = \big({\sum_{i=0}^{d-1} \alpha_{i}{'} \ket{i}\bra{{k_1}+i}}\big) \otimes\big({\sum_{i=0}^{d-1} \beta_{i}{'} \ket{i}\bra{\overline{k_1+k_2}\oplus_q i}}\big)$$ $\otimes \big({\sum_{i=0}^{d-1} \gamma_{i}{'} \ket{i}\bra{d\oplus_r i}}\big)$, where, 
\begin{align*}
\alpha_i=\begin{cases} 1 \hspace{4mm} ;\: 0\leqslant i \leqslant k_1 -1 \\ \frac{1}{\sqrt{2}} \hspace{1mm} ;\: k_1 \leqslant i \leqslant d-1
\end{cases}
\alpha_{i}{'}=\begin{cases} \frac{1}{\sqrt{2}} \hspace{1mm} ;\: 0\leqslant i \leqslant k_2 +k_3 -1 \\ 1 \hspace{4mm} ;\: k_2 +k_3 \leqslant i \leqslant d-1
\end{cases}
\end{align*} \vspace{1mm}
\begin{align*}
\beta_i=\begin{cases} \frac{1}{\sqrt{2}} \hspace{1mm} ;\: 0\leqslant i \leqslant k_1 -1 \\ 1 \hspace{4mm} ;\: k_1 \leqslant i \leqslant k_1+k_2 -1 \\ \frac{1}{\sqrt{2}} \hspace{1mm} ;\: k_1+k_2 \leqslant i \leqslant d -1
\end{cases}
\end{align*}
\begin{align*}
\beta_{i}{'}=\begin{cases} \frac{1}{\sqrt{2}} \hspace{1mm} ;\: 0\leqslant i \leqslant k_3 -1 \\ 1 \hspace{4mm} ;\: k_3 \leqslant i \leqslant k_2+k_3 -1 \\ \frac{1}{\sqrt{2}} \hspace{1mm} ;\: k_2+k_3 \leqslant i \leqslant d -1
\end{cases}
\end{align*} \vspace{1mm}
\begin{align*}
\gamma_i=\begin{cases} \frac{1}{\sqrt{2}} \hspace{1mm} ;\: 0\leqslant i \leqslant k_1+k_2 -1 \\ 1 \hspace{4mm} ;\: k_1+k_2 \leqslant i \leqslant d-1
\end{cases}
\end{align*} 

\begin{align*}\gamma_{i}{'}=\begin{cases} 1 \hspace{4mm} ;\: 0\leqslant i \leqslant k_3 -1 \\ \frac{1}{\sqrt{2}} \hspace{1mm} ;\: k_3 \leqslant i \leqslant d-1
\end{cases}
\end{align*}
(III) Finally for the LOCC case, let $p= d, q= d, r= 2d$ and $$\ket{\psi_1}= \frac{1}{\sqrt{d}}\sum_{i=0}^{d-1}\ket{i,i,i}, \ket{\psi_2}=\frac{1}{\sqrt{d}}\sum_{i=0}^{d-1}\ket{i,i,d+i}.$$ Now, Charlie can perform a local measurement $\{C_1, C_2\}$ on his subsystem to distill the state $\ket{\psi_1}$ from $\rho$ where $C_1= \sum_{i=0}^{d-1}\ket{i}\bra{i}$, $C_2=\sum_{i=0}^{d-1}\ket{i}\bra{d+i}$.\\

The converse part of the theorem is very similar to that of bipartite case. If interested, the reader may find it in the appendix.
\end{proof}
Unlike bipartite scenario, some entanglement can be distilled with unit probability even for the simplest case of a three qubit system by LOCC. However, the entangled state may not be genuine.\\

For example, consider the $2\otimes 2\otimes 2$ mixed state $\rho= p\ket{\psi_1}\bra{\ket{\psi_1}}+(1-p)\ket{\psi_2}\bra{\ket{\psi_2}}$ where, $$\ket{\psi_1}= \frac{1}{\sqrt{2}}(\ket{000}+\ket{111}), ~~\ket{\psi_2}=\frac{1}{\sqrt{2}}(\ket{001}+\ket{110}).$$ First, we write Charlie's system in $\{\ket{+}_{_C},\ket{-}_{_C}\}$ basis where, $$\ket{+}_{_C}=\frac{1}{\sqrt{2}}(\ket{0}_{_C} +\ket{1}_{_C}), ~\ket{-}_{_C}=\frac{1}{\sqrt{2}}(\ket{0}_{_C} -\ket{1}_{_C}).$$ Then, $\ket{\psi_1}$ and $\ket{\psi_2}$ can be written as $$\ket{\psi_1}=\frac{1}{\sqrt{2}}(\ket{\phi^+}_{_{AB}}\ket{+}_{_C} +\ket{\phi^-}_{_{AB}}\ket{-}_{_C}),$$  $$\ket{\psi_2}=\frac{1}{\sqrt{2}}(\ket{\phi^+}_{_{AB}}\ket{+}_{_C} -\ket{\phi^-}_{_{AB}}\ket{-}_{_C})$$ where, $$\ket{\phi^+}_{_{AB}}=\frac{1}{\sqrt{2}}(\ket{00}+\ket{11})_{_{AB}}, ~\ket{\phi^-}_{_{AB}}=\frac{1}{\sqrt{2}}(\ket{00}-\ket{11})_{_{AB}}.$$ Charlie now measures his local system in the basis $\{C_1, C_2\}$ where $C_1=\ket{+}\bra{+}$, $C_2=\ket{-}\bra{-}$. If the outcome is $C_1$, an EPR state $\ket{\phi^+}$ is generated between $A$ and $B$ by tracing out subsystem of Charlie from the system. If $C_2$ occurs, Charlie communicates this result to Bob. He then performs a unitary measurement $\sigma_z$ on his subsystem ($\sigma_z= \ket{0}\bra{0}-\ket{1}\bra{1}$). Tracing out $C$, the EPR $\ket{\phi^+}$ is again generated between $A$ and $B.$ \\

In conclusion, we have provided a bound on dimensions of bipartite systems for which the distillation of pure entangled states with certain Schmidt ranks are never possible and resolves the long standing problem. A similar bound is also observed for tripartite systems when the original state is a mixture of GHZ-type states and the target state is a pure entangled GHZ-type state. Unfortunately, we could not find such bounds for distilling W-type states which we leave as an open problem. Although the mathematical treatment is sometimes very similar to \cite{20,21}, our results provide a complete answer to the bipartite systems and up to some extent for tripartite systems regarding the existence of distillable states. Also the separable operation is optimal for the chosen class of states since it is of Kraus rank two and is effective enough to extract a maximally entangled state for both bipartite and tripartite systems. Moreover, these results prove the existence of a class of both bipartite and tripartite systems where separable operation conclusively outperforms LOCC. This class becomes increasingly large with the increment in local dimensions. The results of tripartite systems can be generalised for $n$-partite systems. Under similar restrictions, for a $n$-partite system $\bigotimes_{i=1}^n a_i$, the bound would be $\sum_{i=1}^n a_i \geqslant(n+1)d$ where the target state is a $d$-level $n$-partite GHZ-type state $\ket{GHZ_{n}^d}= \frac{1}{\sqrt{d}}\sum_{i=0}^{d-1}\ket{i,i,i,......(n\:times)}$. An interesting progress of this work would be if someone will consider W-type states under distillation scenario. Also, it seems that there are only certain classes of mixed states that can be transformed into a pure entangled state. One should naturally be curious to know the characterization of those classes. Moreover, we hope that our work will be helpful in realizing $d$-level entangled state in the distant lab scenario.\\
The authors I.C. and D.S. acknowledge the work as part of Quest initiatives by DST, India. The authors A.B. and I.B. acknowledge the support from UGC, India.

\section*{Appendix : Proof of converse part of Theorem 2}
First we consider the separable case. Let $p+q+r<4d$ where $p=d+k_1,\:q=d+k_2,\:r=d+k_3$ satisfying $k_1+k_2+k_3<d$. If possible, let there exist a mixed state $\rho=p\ket{\psi_1}\bra{\psi_1}+(1-p)\ket{\psi_2}\bra{\psi_2}$ of rank two where $\ket{\psi_1},\ket{\psi_2}$ are two $d$-level (or more) orthogonal GHZ-type states, that can be converted into a pure $d$-level GHZ-type entangled state $\ket{\phi}$ with certainty by some separable operation, say $\{A_i\otimes B_i\otimes C_i\}_i$. Then there exist $A\otimes B\otimes C\in \{A_i\otimes B_i\otimes C_i\}_i$ such that $(A\otimes B\otimes C)\ket{\psi_1}=a\ket{\phi}$ and $(A\otimes B\otimes C)\ket{\psi_2}=b\ket{\phi}$. (i) if both $a, b$ are non-zero, $(A\otimes B\otimes C)(b\ket{\psi_1}-a\ket{\psi_2})= 0$. Since $(A\otimes B\otimes C)(x\ket{\psi_1}+y\ket{\psi_2})=(ax+by)\ket{\phi}\forall x,y\in\mathbb{C}$, it must be that $b\ket{\psi_1}-a\ket{\psi_2}$ is at least a $d$-level GHZ-type state. But the rank space of all $A,B,C$ are $d$-dimensional, implies that the annihilator subspace is atmost of dimension $k_1+k_2+k_3$, a contradiction. (ii) if $b=0, a\neq 0$ then $(A\otimes B\otimes C)\ket{\psi_1}=a\ket{\phi}$ implies that all $A,B,C$ are of rank at least $d$. Thus for $b$ to be zero, we must have $k_1+k_2+k_3\geqslant d$, a contradiction. Therefore in either cases, no separable operation exist.\\
Next for LOCC case, we assume that $max\{k_1,k_2,k_3\}< d$. We again assume that $\rho$ be the mixed state defined above. If possible, let there exist some LOCC protocol $\Omega$ such that $\Omega(\rho)= \ket{\phi}\bra{\phi}$. Any finite round LOCC protocol for deterministic conversion in a tripartite system must be initiated by one of the three parties Alice, Bob or Charlie through a measurement on his (her) subsystem and reporting the outcome to some other party by classical communications. Depending on the outcome, some other party picks a measurement for its part of the system again reports the outcome to some other party and the process continues according to the protocol. These operations therefore create a tree of all possible outcomes where each branch represents a series of operators implemented by the three parties. Let $N_i$ index the $i$-th branch after $N$ rounds of measurement so that $A^{N_i}\otimes B^{N_i}\otimes C^{N_i}$ describes the net operation throughout all $N$ rounds along $N_i$. We now proceed to prove the following claim:\\

For any branch $N_i$, $A^{N_i}\otimes B^{N_i}\otimes C^{N_i}\ket{\psi_1}\neq 0$ if and only if $A^{N_i}\otimes B^{N_i}\otimes C^{N_i}\ket{\psi_2}\neq 0$. By induction, we assume that it is true and let Charlie makes the next measurement. For any $C^j$, let $(N+1)_{ji}$ correspond to the branch $(\mathbb{I}\otimes\mathbb{I}\otimes C^j)(A^{N_i}\otimes B^{N_i}\otimes C^{N_i})$. Since the conversion is deterministic, $A^{N_i}\otimes B^{N_i}\otimes C^{N_i}\ket{\psi_1}$ and $A^{N_i}\otimes B^{N_i}\otimes C^{N_i}\ket{\psi_2}$ are both GHZ-type states of at least level $d$. If $C^j$ is applied on either state, it must at least be of rank $d$ and thus fails to eliminate any $d$-level GHZ-type states. Consequently, $A^{(N+1)_{ji}}\otimes B^{(N+1)_{ji}}\otimes C^{(N+1)_{ji}}\ket{\psi_1}\neq 0$ if and only if  $A^{(N+1)_{ji}}\otimes B^{(N+1)_{ji}}\otimes C^{(N+1)_{ji}}\ket{\psi_2}\neq 0$. This argument is also applicable for the case  $N=0$.\\

Next, we consider a branch in the protocol that converts $\ket{\psi_1}$ to $\ket{\phi}$. The state $\ket{\psi_2}$ must also be converted into $\ket{\phi}$ along this branch. The states $\ket{\psi_1}$ and $\ket{\psi_2}$ are orthogonal but ultimately become the same the state $\ket{\phi}$ in the end of the protocol. So there exist some round $N$ in this branch such that $\ket{\psi_1^{N-1}}$ and $\ket{\psi_2^{N-1}}$ are linearly independent but $\ket{\psi_2^{N}}$ is a scalar multiple of $\ket{\psi_1^{N}}$. Here $\ket{\psi_{1(2)}^{N}}$ is the resultant state after $N$ rounds, which originally was $\ket{\psi_{1(2)}}$. Without loss of generality, we assume that $N$-th round measurement is performed by Charlie. Then $(\mathbb{I}\otimes\mathbb{I}\otimes C)\ket{\psi_1^{N-1}}=a\ket{\phi}$ and $(\mathbb{I}\otimes\mathbb{I}\otimes C)\ket{\psi_2^{N-1}}=b\ket{\phi}$ where $a,b$ are non-zero scalars. Thus $(\mathbb{I}\otimes\mathbb{I}\otimes C)(b\ket{\psi_1^{N-1}}-a\ket{\psi_2^{N-1}})= 0$ where $b\ket{\psi_1^{N-1}}-a\ket{\psi_2^{N-1}}$ is a pure entangled GHZ-type state of level $d$, at least. But this contradicts the fact $k_3< d$. This proves our claim. \hspace{10mm}$\blacksquare$

\begin{thebibliography} {100}
\bibitem{1} C. H. Bennett, G. Brassard, S. Popescu, B. Schumacher, J. Smolin and W. K. Wootters, Phys. Rev. Lett. 76, 722 (1996)
\bibitem{2} M. Horodecki, P. Horodecki, and R. Horodecki,  Phys. Rev. Lett. 80, 5239 (1998)
\bibitem{3} E. Chitambar, D. Leung, L. Man{\v c}inska, M. Ozols, and A. Winter, Everything you always wanted to know about LOCC (but were afraid to ask), Commun. Math. Phys. 328, 303 (2014)
\bibitem{4} B. Schumacher, Phys. Rev. A \textbf{54}, 2614 (1996)
\bibitem{5} B. Groisman and L. Vaidman, J. Phys. A: Math. Gen. 34, 6881 (2001)
\bibitem{6} M. Horodecki, A. Sen(De), U. Sen, and K. Horodecki, Phys. Rev. Lett. 90, 047902 (2003)
\bibitem{7} S. Ghosh, G. Kar, A. Roy, and D. Sarkar, Phys. Rev. A 70, 022304 (2004).
\bibitem{8}  M. Nathanson, J. Math. Phys. 46, 062103 (2005)
\bibitem{9} J. Watrous, Phys. Rev. Lett. 95, 080505 (2005)
\bibitem{10}  J. Niset and N. J. Cerf, Phys. Rev. A 74, 052103 (2006)
\bibitem{11} M.-Y. Ye, W. Jiang, P.-X. Chen, Y.-S. Zhang, Z.-W. Zhou, and G.-C. Guo, Phys. Rev. A 76, 032329 (2007)
\bibitem{12} H. Fan, Phys. Rev. A 75, 014305 (2007)
\bibitem{13} R. Duan, Y. Feng, Z. Ji, and M. Ying, Phys. Rev. Lett. 98, 230502 (2007)
\bibitem{14}  S. Bandyopadhyay and J. Walgate, J. Phys. A: Math. Theor. 42, 072002 (2009)
\bibitem{15} Y. Feng and Y. Shi, IEEE Trans. Inf.
Theory 55, 2799 (2009)
\bibitem{16} N. Yu, R. Duan, and M. Ying, Phys. Rev. Lett. 109, 020506 (2012)
\bibitem{17} Y.-H. Yang, F. Gao, G.-J. Tian, T.-Q. Cao, and Q.-Y.Wen, Phys. Rev. A 88, 024301
(2013)
\bibitem{18}  Z.-C. Zhang, F. Gao, G.-J. Tian, T.-Q. Cao, and Q.-Y.Wen, Phys. Rev. A 90, 022313 (2014)
\bibitem{19}  S. Bandyopadhyay, A. Cosentino, N. Johnston, V. Russo,
J. Watrous, and N. Yu, IEEE Trans. Inf. Theory 61, 3593 (2014)
\bibitem{20} A. Kent, Phys. Rev. Lett. 81, 2839 (1998)
\bibitem{21} E. Chitambar and R. Duan, Phys. Rev. Lett. 103, 110502 (2009)
\end{thebibliography}
\end{document}